\documentclass[11pt]{amsart}

\newtheorem{Theorem}{Theorem}
\newtheorem{Corollary}{Corollary}
\newtheorem{Definition}{Definition}
\newtheorem{Proposition}{Proposition}
\newtheorem{Lemma}{Lemma}

\newcommand{\Id}{{\mathrm{I_d}}}
\newcommand{\dom}{{\mathrm{dom}~}}
\newcommand{\R}{\ensuremath{{\mathrm{I\!R}}}}
\newcommand{\N}{\ensuremath{{\mathrm{I\!N}}}}
\newcommand{\Q}{\ensuremath{{\mathrm{Q\hspace{-2.1mm}\rule{0.3mm}{2.6mm}\;}}}}
\newcommand{\C}{\ensuremath{{\mathrm{C\hspace{-1.7mm}\rule{0.3mm}{2.6mm}\;}}}}
\newcommand{\Z}{\ensuremath{{\mathsf{Z\!\!Z}}}}

\begin{document}

\title[Almost Periodic Orbits and Stability]{Almost Periodic Orbits and
Stability for Quantum Time-Dependent Hamiltonians}
\author{C\'{e}sar R. de Oliveira}
\thanks{CRdeO was partially supported by CNPq (Brazil).}
\address{Departamento de Matem\'{a}tica -- UFSCar, S\~{a}o Carlos, SP,
13560-970 Brazil} \email{oliveira@dm.ufscar.br}
\author{Mariza S. Simsen}
\thanks{MSS was supported by CAPES (Brazil).}
\address{Departamento de Matem\'{a}tica -- UFSCar, S\~{a}o Carlos, SP,
13560-970 Brazil\\} \email{mariza@dm.ufscar.br}

\begin{abstract} We study almost periodic orbits of quantum systems and prove that for
periodic time-dependent Hamiltonians an orbit is almost periodic if, and only if, it is
precompact. In the case of quasiperiodic time-dependence
 we present an example of a precompact orbit that is not almost periodic.
Finally we discuss some simple conditions assuring dynamical stability for
nonautonomous quantum system.
\end{abstract}
\maketitle


\section{Introduction}\label{IntroductionSection} The time
evolution of a quantum mechanical system with time-dependent
Hamiltonians $H(t)$ is determined by the Schr\"{o}dinger equation
\[i\frac{d\psi(t)}{dt}=H(t)\psi(t),\] where $H(t)$ is a family of
self-adjoint operators in the Hilbert space $\mathcal{H}$ and
$\psi(t)\in\mathcal{H}$ for all $t\in\R$. The initial value problem
$\psi(0)=\psi$ has a unique solution
\[\psi(t)\doteq U(t,0)\psi,\] under suitable conditions on $H(t)$
(see \cite{RS,K,K1,I}) and the propagators, or time evolution
operators $U(t,s)$, form a strongly continuous family of unitary
operators acting on $\mathcal{H}$, such that
\[U(t,r)U(r,s)=U(t,s),\qquad \forall r,s,t,\in\R\] \[U(t,t)=
\Id,\qquad \forall t.\] $\Id$ denotes the identity operator. If the Hamiltonian is time-periodic
with period $T$, then $U(t+T,r+T)=U(t,r)$ and the Floquet operator at
$s$ is defined by $U_F(s)\doteq U(s+T,s)$; $U_F(0)$ is simply called 
Floquet operator and denoted by $U_F$, and $U_F(s)$ is unitarily equivalent to $U_F(r)$, $\forall
r,s$. Let
\[
\mathcal{O}(\psi)\doteq\{ U(t,0)\psi:t\in\R\}
\] be
the orbit of a vector $\psi\in\mathcal{H}$.

If $H(t)=H$ is independent of $t$ the time evolution operators are
$U(t,s)=e^{-iH(t-s)}$. In this case, it is a well-known fact that
if $\psi$ is in the point subspace of $H$ then the quantum time
evolution of the state $\psi$, $\psi(t)$, is almost periodic,
since it can be expanded in terms of the eigenfunctions
$\varphi_n$ of $H$, with eigenvalues $E_n$,
\[\psi(t)=\sum_nc_ne^{-iE_nt}\varphi_n.\] Reciprocally, if $\psi(t)$
is almost periodic then using the
results in~\cite{Kat} (Chapter VI) it holds true that $\mathcal{O}(\psi)$ is  precompact and then
$\psi$ is in the point subspace of $H$ (see Theorem~\ref{teo.31} ahead). In this work, we prove that
this fact remains true in the periodic case, that is, $\psi$ is in the point subspace of $U_F$ if,
and only if, $\psi(t)$ is almost periodic (see Theorem~\ref{teo.31a}).

In the studies of time-dependent systems it is common to consider the
quasienergy operator, i.e., a self-adjoint operator formally given
by \[K=-i\frac{d}{dt}+H(t)\] acting in some enlarged Hilbert
space. The quasienergy operator $K$ was previously defined for
periodic Hamiltonians~\cite{Y,H} and then generalized for general
time dependence in~\cite{H1}. In the periodic case it was proved
that \[e^{-iKT}\simeq \Id\otimes U_F,\] where $\simeq$ means
unitary equivalence.

A natural framework for considering general time-dependent
perturbations, which includes both periodic and the random
potentials as special cases, is to write $H(t)$ in the form
\[H(t)=H(g_t(\theta))=H_0+V(g_t(\theta)),\] where $g_t:\Omega\rightarrow\Omega$ is
an invertible flow on a compact manifold $\Omega$ with a probability ergodic
 measure $\mu$ and $H_0$ is the Hamiltonian of the
isolated system (see~\cite{JL,BJL}). Again, under suitable
conditions on $V$ there exists a unitary time evolution operator
$U_{\theta}(t,s)$ and the generalized quasienergy operator is
defined~\cite{JL} on $L^2(\Omega,\mathcal{H},d\mu)$ by
\[(e^{-i\tilde{K}t}f)_{\theta}=\mathcal{F}_{-t}U_{\theta}(t,0)f_{\theta}=
U_{\theta}(0,-t)\mathcal{F}_{-t}f_{\theta},\] where
$\mathcal{F}_{-t}f_{\theta}=f_{g_{-t}(\theta)};$ we refer to this construction as {\it
Jauslin-Lebowitz formulation.} The operator
$\tilde{K}$ acts as
\[(\tilde{K}f)_{\theta}=i\frac{d}{dt}f_{g_{-t}(\theta)}\Big|_{t=0}+H_{\theta}f_{\theta}.\]
In the case of a periodic potential one has $\Omega=S^1\equiv[0,2\pi)$,
$g_t(\theta)=\theta+\omega t$ and $d\mu=\frac{d\theta}{2\pi}$.

For quasiperiodic potentials with two incommensurate
frequencies $\omega_1/\omega_2\notin\Q$ the manifold $\Omega$ is
$S^1\times S^1$,
$g_t(\theta_1,\theta_2)=(\theta_1+\omega_1t,\theta_2+\omega_2t)$
and $d\mu=\frac{d\theta_1}{2\pi}\frac{d\theta_2}{2\pi}$. We denote
the two periods by $T_j=\frac{2\pi}{\omega_j}$. In this case the
generalized Floquet operator acting on $\mathcal{K}_1\doteq
L^2(S^1,\mathcal{H},\frac{d\theta_1}{2\pi})$ is defined by
\begin{equation}\label{FloquetGenEq}
U_{\mathrm{F}}=\mathcal{T}_{-T_2}u_1,
\end{equation} where $u_1(\theta_1)=U_{(\theta_1,0)}(T_2,0)$
($\doteq$ monodromy operator) and
$(\mathcal{T}_{-T_2}\phi)(\theta_1)=\phi(\theta_1-\omega_1T_2)$.

Let $A:\dom A\subset\mathcal{H}\rightarrow\mathcal{H}$ be an
unbounded positive self-adjoint operator with discrete spectrum
which we call a {\it probe operator}. Assuming that if $\psi\in\dom
A$, then $U(t,0)\psi\in\dom A$ for all $t\geq0$, a very
interesting question is about the behavior of the expectation
value of $A$, that is, 
\[
E_{\psi}^{A}(t)\equiv\langle
U(t,0)\psi,AU(t,0)\psi\rangle.
\] We say the system is
$A$-dynamically stable if $E_{\psi}^{A}(t)$ is a bounded function
of time, and $A$-dynamically unstable otherwise. A particular case
is when the Hamiltonian has the form $H(t)=H_0+V(t)$ and $A=H_0$.
In this work we discuss some simple conditions assuring dynamical
stability, mainly when either the Floquet or quasienergy operator
has purely point spectrum; recall that in the periodic case it is
known  that continuous spectrum of the Floquet operator implies
dynamical instability (see Section~\ref{PreliminarSection}).

Usually it is not a simple task to get results on dynamical (in)stability  in the
original Hilbert space $H$ through properties of $K$ or
$\tilde{K}$ acting in the corresponding enlarged space. We present
some theoretical results about this point in
Section~\ref{BoundedSection}. An important result in the periodic case was proved
in~\cite{DSSV}, i.e., that the applicability of the KAM method for the
quasienergy operator $K$, which is a technique to find out a unitary
operator $U$ such that $UKU^{-1}=D,$ where $D$ is pure point, gives a
uniform bound at the expectation value of the energy for a class
of time-periodic Hamiltonians of the form $H(t)=H_0+V(t)$
considered in~\cite{DLSV}. 

The study of precompacity (and related properties) of orbits of a time-dependent quantum system and
their connection with  spectral type and stability  was carried out, e.g.,
in~\cite{EV,dOT,dO,BJLPN,JL,BJL}. In this work we prove that in the periodic case (including the
autonomous case) the orbit
$\mathcal{O}(\psi)$ is precompact if, and only if, $\psi(t)$ is an almost periodic
function. Moreover, already in the quasiperiodic case we present an example with precompact
orbits which are not almost periodic.

This paper is organized as follows. In
Section~\ref{PreliminarSection} we recall some subspaces of
$\mathcal{H}$ that were studied in the literature and the results
that connect this subspaces with dynamical (in)stability  and
spectral properties of the Floquet or quasienergy operators. In
Section~\ref{PeneperiodicSection} we present ours results about
almost periodic orbits. In Section~\ref{BoundedSection} we discuss
some simple conditions assuring dynamical stability; we pay special attention to connection between
enlarged spaces and the original quantum Hilbert space. A number of known results are recalled in
the text in order to make it as readable as possible.

\

\section{Preliminaries}\label{PreliminarSection} In this section we present a short account of
suitable subspaces and relations among them, in order to put our results in context.

Consider a time-dependent Hamiltonian $H(t)$ acting in a separable
Hilbert space $\mathcal{H}$, which may be nonperiodic,
and let $U(t,0)$ the corresponding propagators. Denote by
$A:\dom A\subset\mathcal{H}\rightarrow\mathcal{H}$  a probe
operator, such that $\dom A$ is invariant under time evolution $U(t,0)$.
Let $F(A>E)$ be the spectral projection onto the closed space
spanned by the eigenvectors of $A$ corresponding to the
eigenvalues larger than $E\in\R$. The relevant definitions are as
follows~\cite{EV,dOT,dO,BJLPN}.

\

\begin{Definition}\label{def.26} (i) $\mathcal{H}_{pc}\doteq\{\xi\in\mathcal{H}:
\mathcal{O}(\xi)\;\textrm{is precompact in}\;\mathcal{H}\}$.
\newline (ii)
$\mathcal{H}_{\mathrm{f}}\doteq\big\{\xi\in\mathcal{H}:
\lim_{\tau\rightarrow\infty}\frac{1}{\tau}
\int_0^{\tau}\|CU(t,0)\xi\|dt=0\;\textrm{for any compact}$\break$
\textrm{operator}\; C\big\}$.
\newline (iii)
$\mathcal{H}_{\mathrm{be}}\doteq\{0\neq\xi\in\mathcal{H}:
\lim_{E\rightarrow\infty}\sup_{t\in\R}
\|F(A>E)U(t,0)\frac{\xi}{\|\xi\|}\|=0\}\cup\{0\}$. \newline (iv)
$\mathcal{H}_{\mathrm{ue}}\doteq\{0\neq\xi\in\mathcal{H}:
\lim_{E\rightarrow\infty}\sup_{t\in\R}
\|F(A>E)U(t,0)\frac{\xi}{\|\xi\|}\|=1\}\cup\{0\}$. \newline (v)
$\mathcal{S}^{\mathrm{bd}}(A)\doteq\{\xi\in\dom A:\text{the
function}\;t\mapsto E_{\xi}^A(t)\;\textrm{is bounded}\}$.
\newline (vi)
$\mathcal{S}^{\mathrm{un}}(A)\doteq\{\xi\in\dom A:\text{the
function}\;t\mapsto E_{\xi}^A(t)\;\textrm{is unbounded}\}$.
\end{Definition}

Important compact operators are the projections onto finite
subspaces of $\mathcal{H}$, so that the elements of
$\mathcal{H}_{\mathrm{f}}$ are interpreted as the vectors that
under time evolution leave, on average, any finite-dimensional
subspace of $\mathcal{H}$.

Some basic properties of the sets that appeared in the above
definition are summarized ahead. For proofs we refer the reader
to~\cite{dO,dOT,EV,BJLPN}.

\

\begin{Theorem}\label{GeneralTheo} Let $H(t)$ be a time-dependent
Hamiltonian and $A$ as above; then: \newline (a)
$\mathcal{H}_{\mathrm{f}}$ and $\mathcal{H}_{\mathrm{pc}}$ are
closed subspaces of $\mathcal{H}$.
\newline (b)
$\mathcal{H}_{\mathrm{pc}}\perp\mathcal{H}_{\mathrm{f}}$. \newline
(c) $\mathcal{H}_{\mathrm{be}}=\mathcal{H}_{\mathrm{pc}}$ and
$\mathcal{H}_{\mathrm{f}}\subset\mathcal{H}_{\mathrm{ue}}$.
\newline (d) If $\xi\in\dom A$ and
$\xi\notin\mathcal{H}_{\mathrm{pc}}$ then
$\xi\in\mathcal{S}^{\mathrm{un}}(A)$, that is,
$\mathcal{S}^{\mathrm{bd}}(A)\subset\mathcal{H}_{\mathrm{pc}}$. In
particular, $(\dom A\cap\mathcal{H}_{\mathrm{f}})\setminus\{0\}
\subset\mathcal{S}^{\mathrm{un}}(A)$.
\end{Theorem}

Note that if the Hamiltonian $H(t)$ has the form
$H(t)=H_0+V(t)$ with $H_0$ an unbounded, positive, self-adjoint
operator with discrete spectrum, then Theorem~\ref{GeneralTheo}(d) holds
true for $A=H_0$.

\

\subsection{Periodic Case}\label{PeriodicSection} If $H(t)$ is periodic
of period $T$ and $U_F=U(T,0)$ is the corresponding Floquet
operator, we denote by $\mathcal{H}_{\mathrm{p}}$ the point
spectral subspace and by $\mathcal{H}_{\mathrm{c}}$ the continuous
subspace of the Floquet operator $U_{\mathrm{F}}$. Recall the important

\

\begin{Theorem}[RAGE]\label{teo.31R} Let
$C:\mathcal{H}\rightarrow\mathcal{H}$ be a compact operator and
$\xi\in\mathcal{H}_{\mathrm{c}}$, then
\[\lim_{\tau\rightarrow\infty}\frac{1}{\tau}\int_0^{\tau}\|CU(t,0)\xi\|dt=0.\]
\end{Theorem}

\

A detailed proof of Theorem~\ref{teo.31R} can be found in
\cite{EV}; this result was firstly proved for the autonomous case
(see, e.g., \cite{AG}). As a consequence of this theorem it follows that if
$\xi\in\mathcal{H}_{\mathrm{c}}$ then
$\xi\in\mathcal{H}_{\mathrm{f}}$, so by
Theorem~\ref{GeneralTheo}(d) it follows that $\langle
U(t,0)\xi,AU(t,0)\xi\rangle$ is unbounded. Thus, as it is well known, the presence of
continuous spectrum for the Floquet operator is a signature of
 quantum instability. In principle, one would expect that  a Floquet operator with
purely point spectrum would imply quantum stability, however there are
examples with purely point spectrum and dynamically
unstable; see \cite{dRJLS,JSS,dOP} for examples in the autonomous case and
\cite{dOS} for the time-periodic case.

Using the above theorem  and a series of technical lemmas in~\cite{dOT}, one gets

\

\begin{Theorem}\label{teo.31} If the Hamiltonian operator is periodic in time, then
\newline (a) $\mathcal{H}_{\mathrm{p}}=\mathcal{H}_{\mathrm{be}}=
\mathcal{H}_{\mathrm{pc}}$;
\newline (b) $\mathcal{H}_{\mathrm{c}}=\mathcal{H}_{\mathrm{ue}}=\mathcal{H}_{\mathrm{f}}$.
\end{Theorem}

\

We observe that Theorem~\ref{teo.31} also holds in the autonomous
case $H(t)=H$ and with $\mathcal{H}_{\mathrm{p}}$ and
$\mathcal{H}_{\mathrm{c}}$ denoting, respectively, the point and
continuous subspace of the Hamiltonian $H$.

According to the above-quoted results, for periodic systems we
have
\begin{equation}\label{eq.22}
\mathcal{H}=\mathcal{H}_{\mathrm{pc}}\oplus\mathcal{H}_{\mathrm{f}}.
\end{equation}
In \cite{dO} was presented an example for which relation
(\ref{eq.22}) does not hold for nonperiodic time dependence. It
was defined the ``unusual" subspace $\mathcal{H}_{\mathrm{a}}$ by
the relation
\[\mathcal{H}=\mathcal{H}_{\mathrm{pc}}\oplus\mathcal{H}_{\mathrm{f}}
\oplus\mathcal{H}_{\mathrm{a}},\] and constructed a nonperiodic
Hamiltonian such that $\mathcal{H}=\mathcal{H}_{\mathrm{a}}$. The
example is given by the Floquet operator generated by the kicked Hamiltonian
\[H(t)=p^2+x\sum_{n=1}^{\infty}\epsilon_n\delta(t-n),\qquad
x\in[0,2\pi),\] acting on $\mathcal{H}=L^2(\mathbb{T})$ and
$\epsilon_n\in\{-1,0,1\}$ adequately chosen.  This example
illustrates some possible unusual properties of nonstationary
quantum systems.

\

\subsection{Quasiperiodic Case}\label{QuasiperiodicSection} In this case we have the
generalized Floquet operator $U_{\mathrm{F}}$ as defined in
(\ref{FloquetGenEq}), acting on the enlarged space
$\mathcal{K}_1=L^2(S^1,\mathcal{H},\frac{d\theta_1}{2\pi})$, and
the generalized quasienergy operator $\tilde{K}$ acting in
$L^2(S^1\times
S^1,\mathcal{H},\frac{d\theta_1}{2\pi}\frac{d\theta_2}{2\pi})$. We
denote, respectively, by $\mathcal{K}_{1,p}$ and
$\mathcal{K}_{1,\mathrm{c}}$ the point and continuous subspace of
the generalized Floquet operator~$U_{\mathrm{F}}$.

For each fixed $t$ let the unitary operator
$U(t):\mathcal{K}_1\rightarrow\mathcal{K}_1$ be given by
$(U(t)\psi)(\theta_1)=U_{(\theta_1,0)}(t,0)\psi(\theta_1)$, that
is,
\[
U(t)=\int_{S^1}^{\oplus}U_{(\theta_1,0)}(t,0)\frac{d\theta_1}{2\pi},
\]
and given $\psi\in\mathcal{K}_1$ let
$\tilde{\mathcal{O}}(\psi)=\{U(t)\psi:t\in\R\}$ be the orbit
of $\psi$ in the enlarged space $\mathcal{K}_1$.

Let $A:\dom A\subset\mathcal{K}_1\rightarrow\mathcal{K}_1$ be a probe
operator with $U(t)\dom A\subset\dom A$ and $F(A>E)$ as before.
The relevant definitions are as follows~\cite{JL,BJL,dOT}:

\

\begin{Definition}\label{def.33} (a) $\mathcal{K}_{1,\mathrm{f}}\doteq\big\{\psi\in
\mathcal{K}_1:\lim_{\tau\rightarrow\infty}\frac{1}{\tau}\int_0^{\tau}
\|CU(t)\psi\|_{\mathcal{K}_1}dt=0\;\textrm{for any}$
$\textrm{compact operator}\; C\;\textrm{in}\;\mathcal{K}_1\big\}$.
\newline (b)
$\mathcal{K}_{1,\mathrm{pc}}=\{\psi\in\mathcal{K}_1:\tilde{\mathcal{O}}(\psi)\;
\textrm{is precompact in}\;\mathcal{K}_1\}$. \newline (c)
$\mathcal{K}_{1,\mathrm{be}}\doteq\{0\neq\psi\in\mathcal{K}_1:
\lim_{E\rightarrow\infty}\sup_{t\in\R}
\|F(A>E)U(t)\frac{\psi}{\|\psi\|}\|=0\}\cup\{0\}$. \newline (d)
$\mathcal{K}_{1,\mathrm{ue}}\doteq\{0\neq\psi\in\mathcal{K}_1:
\lim_{E\rightarrow\infty}\sup_{t\in\R}
\|F(A>E)U(t)\frac{\psi}{\|\psi\|}\|=1\}\cup\{0\}$.
\end{Definition}

In~\cite{JL} it was proved the analog of the RAGE Theorem for the
quasiperiodic case. The proof is an adaptation of the similar
statement in the periodic case discussed in~\cite{EV}. As in
the periodic case one has:

\

\begin{Theorem}\label{teo.35} If the Hamiltonian operator is quasiperiodic
in time, then \newline (a)
$\mathcal{K}_{1,p}=\mathcal{K}_{1,\mathrm{pc}}=\mathcal{K}_{1,\mathrm{be}}$;
\newline (b)
$\mathcal{K}_{1,\mathrm{c}}=\mathcal{K}_{1,\mathrm{ue}}=\mathcal{K}_{1,\mathrm{f}}$.
\end{Theorem}

\

 It is worth
mentioning that the relation between the energy growth and the
characterizations in Definition~\ref{def.33} is not as direct as
in the case of periodic and autonomous potentials. The above theorem
holds on the enlarged space $\mathcal{K}_1$ so that a generalized
operator with continuous spectrum does not ensure unbounded energy
growth in the original Hilbert space $\mathcal{H}$, although it
does in $\mathcal{K}_1$. See~\cite{JL,BJL} for interesting
examples on systems with time-quasiperiodic dependence.

\

\section{Almost Periodic Orbits}\label{PeneperiodicSection} Let
$\mathcal{B}$ be a Banach space. A continuous function
$f:\R\rightarrow\mathcal{B}$ is called \textit{almost periodic}
if for any number $\epsilon>0$, one can find a number
$l(\epsilon)>0$ such that any interval of the real line of length
$l(\epsilon)$ contains at least one point $\tau$ with the property
that \[\|f(t+\tau)-f(t)\|<\epsilon, \qquad\qquad \forall t\in\R.\]
For properties of almost periodic functions we refer the reader
to~\cite{COR,Kat}. Now we introduce the following subset of
$\mathcal{H}$:
\[\mathcal{H}_{\mathrm{ap}}\doteq\{\xi\in\mathcal{H}:
\textrm{the function}\;\R\ni t\mapsto\xi(t)=U(t,0)\xi\;\textrm{is
almost periodic}\}.\]
By abuse of language sometimes we say that the orbit $\mathcal{O}(\xi)$ is almost periodic.

For general time dependence one has

\

\begin{Proposition}\label{prop.27a} $\mathcal{H}_{\mathrm{ap}}$ is a closed subspace
of $\mathcal{H}$ and
$\mathcal{H}_{\mathrm{ap}}\subset\mathcal{H}_{\mathrm{pc}}$.
\end{Proposition}
\begin{proof} Clearly $0\in\mathcal{H}_{\mathrm{ap}}$.
If $\xi,\psi\in\mathcal{H}_{\mathrm{ap}}$ then
$\xi(t)=U(t,0)\xi$ and $\psi(t)=U(t,0)\psi$ are almost periodic
functions. Since the sum of two almost periodic functions with
values in $\mathcal{H}$ is an almost periodic function, it follows
that
$\xi(t)+\psi(t)=U(t,0)\xi+U(t,0)\psi=U(t,0)(\xi+\psi)=(\xi+\psi)(t)$
is an almost periodic function. So
$\xi+\psi\in\mathcal{H}_{\mathrm{ap}}$. Now, let
$\xi\in\mathcal{H}_{\mathrm{ap}}$ and $\lambda$ a complex
number, then $\xi(t)=U(t,0)\xi$ is an almost periodic function.
Since $\lambda\xi(t)=\lambda U(t,0)\xi=U(t,0)(\lambda\xi)$
is an almost periodic function, it follows that
$\lambda\xi\in\mathcal{H}_{\mathrm{ap}}$. So
$\mathcal{H}_{\mathrm{ap}}$ is a vector subspace of
$\mathcal{H}$.

Suppose that $\{\xi_j\}\subset\mathcal{H}_{\mathrm{ap}}$ and
$\lim_{j\rightarrow\infty}\xi_j=\xi$. Given $\epsilon>0$ there
exists $N\in\N$ such that $\|\xi_j-\xi\|<\epsilon$ for all $j\geq
N$; thus, there exists $N$ as above such that
$j\geq N$ implies that $\forall\;t\in\R$
\[\|\xi(t)-\xi_j(t)\|=\|U(t,0)\xi-U(t,0)\xi_j\|\leq\|\xi-\xi_j\|<\epsilon.\]
So $\xi_j(t)\rightarrow\xi(t)$ uniformly in $\R$ in the sense of
convergence in the norm. Since each $\xi_j(t)$ is an almost
periodic function, it follows that $\xi(t)$ is an almost periodic
function (Theorem 6.4 in~\cite{COR}) and
$\xi\in\mathcal{H}_{\mathrm{ap}}$, which shows that
$\mathcal{H}_{\mathrm{ap}}$ is a closed vector subspace of
$\mathcal{H}$.

Since the set of values of an almost periodic function with values
in $\mathcal{H}$ is precompact in $\mathcal{H}$ (Theorem 6.5
in~\cite{COR}), it follows that
$\mathcal{H}_{\mathrm{ap}}\subset\mathcal{H}_{\mathrm{pc}}$.
\end{proof}

\subsection{Periodic Systems}

If the Hamiltonian time dependence is periodic (or autonomous) more can be said.

\begin{Proposition}\label{EigenAlmostProp} If the Hamiltonian
operator is periodic in time and $\xi\in\mathcal{H}_{\mathrm{p}}$
is an eigenvector of $U_{\mathrm{F}}$, that is,
$U_{\mathrm{F}}\xi=e^{-i\alpha}\xi$, $\alpha\in\R$, then
$\xi\in\mathcal{H}_{\mathrm{ap}}\subset\mathcal{H}_{\mathrm{pc}}$.
\end{Proposition}
\begin{proof} Since $U(t,0)$ is strongly continuous
the map $t\mapsto\xi(t)$ is continuous.

Any $t\in\R$ can be written in the form $t=nT+s$, with $n\in\Z$
and $0\leq s<T$. We have $U_{\mathrm{F}}\xi=e^{-i\alpha}\xi$ and
$U_{\mathrm{F}}^{-1}\xi=e^{i\alpha}\xi$. Since for $t\geq0$
($n\geq0$)
\begin{eqnarray*}
U(t,0)\xi &=& U(s+nT,nT)U(nT,(n-1)T)\ldots U(T,0)\xi\\ &=&
U(s,0)\underbrace{U(T,0)\ldots
U(T,0)}_{n\;\textrm{factors}}\xi=U(s,0)e^{-in\alpha}\xi,
\end{eqnarray*}
and for $t<0$ ($n<0$)
\begin{eqnarray*}
U(t,0)\xi &=& U(s+nT,nT)U(nT,(n+1)T)\ldots U(-T,0)\xi\\ &=&
U(s,0)\underbrace{U(T,0)^{-1}\ldots
U(T,0)^{-1}}_{n\;\textrm{factors}}\xi=U(s,0)e^{-in\alpha}\xi,
\end{eqnarray*}
it follows that
\[U(t,0)\xi=U(s,0)e^{-in\alpha}\xi,\] for $t=nT+s\in\R$,
$n\in\Z$ and $0\leq s<T$. So for each $t=nT+s\in\R$
\begin{eqnarray*}
\xi(t+T)&=&U(t+T,0)\xi=U(s,0)e^{-i(n+1)\alpha}\xi\\&=&
e^{-i\alpha}U(s,0)e^{-in\alpha}\xi=e^{-i\alpha}\xi
U(t,0)\xi=e^{-i\alpha}\xi(t),
\end{eqnarray*}
so $t\rightarrow\xi(t)$ is an almost periodic function and the
result is proved.
\end{proof}

Summing up, we conclude:

\

\begin{Theorem}\label{teo.31a} If the Hamiltonian operator is
periodic in time, then \newline (a)
$\mathcal{H}_{\mathrm{p}}=\mathcal{H}_{\mathrm{be}}=
\mathcal{H}_{\mathrm{pc}}=\mathcal{H}_{\mathrm{ap}}$;
\newline (b) $\mathcal{H}_{\mathrm{c}}=\mathcal{H}_{\mathrm{ue}}=
\mathcal{H}_{\mathrm{f}}$.
\end{Theorem}
\begin{proof}
It is enough to prove that
$\mathcal{H}_{\mathrm{pc}}=\mathcal{H}_{\mathrm{ap}}$. The inclusion
$\mathcal{H}_{\mathrm{ap}}\subset\mathcal{H}_{\mathrm{pc}}$ was proved in Proposition~\ref{prop.27a}.
On the other hand, it is a consequence of Propositions~~\ref{prop.27a} and~\ref{EigenAlmostProp} 
that
$\mathcal{H}_{\mathrm{p}}\subset\mathcal{H}_{\mathrm{ap}}$. Since
$\mathcal{H}_{\mathrm{p}}=\mathcal{H}_{\mathrm{pc}}$, it follows that
$\mathcal{H}_{\mathrm{pc}}\subset\mathcal{H}_{\mathrm{ap}}$.
\end{proof}

Theorem~\ref{teo.31a} holds also for autonomous Hamiltonians.

\subsection{Quasiperiodic Systems}
In the above theorem we proved that for time-periodic Hamiltonians an orbit $\mathcal{O}(\xi)$ is
precompact if, and only if, $t\mapsto\xi(t)$ is almost periodic. Now we construct an example showing
that already in the case of time-quasiperiodic Hamiltonians there are precompact orbits
that are not almost periodic.

\

\noindent{\bf Example} Given the matrix
\[
u_1(\theta_1)=
\left(\begin{array}{cc} e^{i\theta_1} & 0\\ 0 &
e^{-i\theta_1}\end{array}\right),
\] it is known (see Lemma 5.1 in~\cite{BJL})
that there exists a quasiperiodic Hamiltonian $H_{\theta}(t)$, $\theta=(\theta_1,\theta_2)$, acting
on
$\mathcal{H}=\C^2$, of the form
\begin{equation}\label{eq.25}
H_{\theta}(t)=h_0(t)\Id+\sum_{j=1}^3h_j(t)\sigma_j,
\end{equation}
where $\sigma_j$ are the Pauli matrices, and $h_j(t)$ are real
quasiperiodic functions, i.e., $h_j(t)=
\bar{h}_j(\omega_1t+\theta_1,\omega_2t+\theta_2)$, where
$\bar{h}_j(\theta_1,\theta_2)$ are continuous and $2\pi$-periodic
in the two arguments $\theta_1,\;\theta_2\in
S^1,$ and $\omega_1,\;\omega_2$ are positive real numbers 
so that $u_1(\theta_1)=U_{(\theta_1,0)}(T_2,0)$ is the
corresponding monodromy operator. Moreover, the corresponding
generalized Floquet operator $U_F=\mathcal{T}_{-T_2}u_1$ has
absolutely continuous spectrum for any irrational
$\alpha\doteq\frac{\omega_1}{\omega_2}$.

By the construction in
the proof of  Lemma~5.1 in~\cite{BJL}, it is found that for $k\in\Z$,
$k>0$,
\begin{eqnarray*}
U_{(\theta_1,0)}(kT_2,0) &=& u_1(\theta_1+(k-1)2\pi\alpha)\ldots
u_1(\theta_1+2\pi\alpha)u_1(\theta_1)\\ &=&
\left(\begin{array}{cc} e^{i(\theta_1+(k-1)2\pi\alpha)}& 0 \\ 0 &
e^{-i(\theta_1+(k-1)2\pi\alpha)}\\
\end{array}\right)\ldots\left(\begin{array}{cc}
e^{i\theta_1}& 0 \\ 0 & e^{-i\theta_1}\\
\end{array}\right)\\ &=& \left(\begin{array}{cc}
e^{i(\theta_1+(k-1)2\pi\alpha)}\ldots e^{i\theta_1} & 0
\\ 0 &
e^{-i(\theta_1+(k-1)2\pi\alpha)}\ldots e^{-i\theta_1}\\
\end{array}\right)\\ &=& \left(\begin{array}{cc}
e^{i(k\theta_1+(1+2+\ldots(k-1))2\pi\alpha)} & 0
\\ 0 & e^{-i(k\theta_1+(1+2+\ldots(k-1))2\pi\alpha)} \\
\end{array}\right)\\ &=& \left(\begin{array}{cc}
e^{ik(\theta_1+(k-1)\pi\alpha)} & 0
\\ 0 & e^{-ik(\theta_1+(k-1)\pi\alpha)} \\
\end{array}\right);
\end{eqnarray*}
for $k<0$ the same expression is found. Therefore, for all $k\in\Z$
\[U_{(\theta_1,0)}(kT_2,0)=\left(\begin{array}{cc}
e^{ik(\theta_1+(k-1)\pi\alpha)} & 0
\\ 0 & e^{-ik(\theta_1+(k-1)\pi\alpha)} \\
\end{array}\right).\] Moreover, for $\theta_1\in S^1$, $0\leq t\leq T_2$, define
\[v(t;\theta_1)=\left(\begin{array}{cc}
e^{i\frac{t}{T_2}(\theta_1+(\frac{t}{T_2}-1)\pi\alpha)} & 0
\\ 0 & e^{-i\frac{t}{T_2}(\theta_1+(\frac{t}{T_2}-1)\pi\alpha)} \\
\end{array}\right),\] which  is differentiable with respect to $t$ and satisfies
\[v(0;\theta_1)=\left(\begin{array}{cc}
1 & 0 \\ 0 & 1 \\ \end{array}\right)=\Id,\qquad
v(T_2;\theta_1)=\left(\begin{array}{cc} e^{i\theta_1} & 0 \\ 0 & e^{-i\theta_1} \\
\end{array}\right)=u_1(\theta_1).\] So for $t\in\R$, $t=kT_2+\delta_t$,
$0\leq\delta_t\leq T_2,$ one has
\[U_{(\theta_1,0)}(t,0)=v(\delta_t;\theta_1+k2\pi\alpha)U_{(\theta_1,0)}(kT_2,0).\]

Therefore, for $\xi\in\mathcal{H}=\C^2$, $\xi=\left(\begin{array}{cc}
\xi_1 \\ \xi_2 \\
\end{array}\right),$ we have \[U_{(\theta_1,0)}(t,0)\xi=\left(\begin{array}{cc}
e^{i\frac{t}{T_2}(\theta_1+(\frac{t}{T_2}-1)\pi\alpha)}\xi_1 \\
e^{-i\frac{t}{T_2}(\theta_1+(\frac{t}{T_2}-1)\pi\alpha)}\xi_2 \\
\end{array}\right).\] Since the map, for $0\neq
a\in\R$, $t\mapsto\sin at^2$ is not almost periodic, because
it is not uniformly continuous, we conclude that the
map $t\mapsto e^{iat^2}$ is not almost periodic.
Thus,  \[t\mapsto
g(t)=e^{i\frac{t}{T_2}(\theta_1+(\frac{t}{T_2}-1)\pi\alpha)}=e^{i\frac{t}{T_2}\theta_1}
e^{it^2\frac{\omega_1\omega_2}{4\pi}}e^{-it\frac{\omega_1}{2}}\] is not almost periodic,
because on the contrary the
map
\[
e^{-i\frac{t}{T_2}\theta_1}g(t)e^{it\frac{\omega_1}{2}}=
e^{it^2\frac{\omega_1\omega_2}{4\pi}}
\] would be almost periodic.

Therefore, if $\xi\neq0$ then the map $t\mapsto
U_{(\theta_1,0)}(t,0)\xi$ is not almost periodic for all
$\theta_1\in S^1$. Hence we have got an example of precompact orbit (a
closed and bounded set on $\C^2$ is compact) which is not almost periodic. This finishes the example.

\

The above example  can be extend to the infinite dimensional Hilbert space
$\mathcal{H}=\bigoplus_{n\in\N}\C^2$ of the elements
$\xi=(\xi_n)_{n\in\N}$ with $\xi_n\in\C^2$ and
$\sum_n|\xi_n|^2<\infty$. Denote
\[\tilde{u_1}(\theta_1)=\left(\begin{array}{cc} e^{i\theta_1} & 0\\
0 & e^{-i\theta_1}\\
\end{array}\right);\] we know that there exists a quasiperiodic
$\tilde{H}_{\theta}(t)$ such that $\tilde{u_1}(\theta_1)$ is the
corresponding monodromy operator. Moreover,
$\sigma(\tilde{U_{\mathrm{F}}})$ is absolutely continuous for all
irrational $\alpha$. Let
\[u_1(\theta_1)=\left(\begin{array}{cccc}
\left(\begin{array}{cc} e^{i\theta_1} & 0\\
0 & e^{-i\theta_1}\\
\end{array}\right) &  &  &  \\
  & \left(\begin{array}{cc} e^{i\theta_1} & 0\\
0 & e^{-i\theta_1}\\
\end{array}\right) & &  \\
 & & \left(\begin{array}{cc} e^{i\theta_1} & 0\\
0 & e^{-i\theta_1}\\
\end{array}\right) & \\
 & & & \ddots \\
\end{array}\right)\] or, writing in the another way,
$u_1(\theta_1)=\bigoplus\tilde{u_1}(\theta_1)$. For
$\xi\in\mathcal{H}$ one has
$u_1(\theta_1)\xi=\bigoplus\tilde{u_1}(\theta_1)\xi_n$. The
Floquet operator corresponding to $u_1(\theta_1)$,
$U_{\mathrm{F}}=\mathcal{T}_{-T_2}u_1:L^2(S^1,\mathcal{H},
\frac{d\theta_1}{2\pi})\rightarrow
L^2(S^1,\mathcal{H},\frac{d\theta_1}{2\pi})$ has absolutely
continuous spectrum for all irrational $\alpha$.

If $H_{\theta}(t)=\bigoplus_{n\in\N}\tilde{H}_{\theta}(t)$ then
the propagator of $H_{\theta}(t)$ is
$U_{\theta}(t,0)=\break\bigoplus\tilde{U}_{\theta}(t,0)$. Thus,
$H_{\theta}(t)$ has $u_1(\theta_1)$ as the corresponding monodromy
operator, and given $0\neq\xi\in\mathcal{H}$ and
$\theta=(\theta_1,0)\in S^1\times S^1$ one has
\begin{eqnarray*}
U_{(\theta_1,0)}(t,0)\xi &=& \bigoplus_n\left(\begin{array}{cc}
e^{i\frac{t}{T_2}(\theta_1+(\frac{t}{T_2}-1)\pi\alpha)} & 0\\
0 & e^{-i\frac{t}{T_2}(\theta_1+(\frac{t}{T_2}-1)\pi\alpha)}\\
\end{array}\right) \xi_n \\ &=&
\bigoplus_n\left(\begin{array}{c}
e^{i\frac{t}{T_2}(\theta_1+(\frac{t}{T_2}-1)\pi\alpha)}\xi_n^1\\
e^{-i\frac{t}{T_2}(\theta_1+(\frac{t}{T_2}-1)\pi\alpha)}\xi_n^2\\
\end{array}\right).
\end{eqnarray*}
So $t\mapsto U_{(\theta_1,0)}(t,0)\xi$ is not almost
periodic. If $\xi$ satisfies $\xi=\oplus\xi_n$ with $\xi_n\neq0$
if, and only if, $n=l$, then
\[U_{(\theta_1,0)}(t,0)\xi=\left(\begin{array}{c}
e^{i\frac{t}{T_2}(\theta_1+(\frac{t}{T_2}-1)\pi\alpha)}\xi_l^1\\
e^{-i\frac{t}{T_2}(\theta_1+(\frac{t}{T_2}-1)\pi\alpha)}\xi_l^2\\
\end{array}\right),\] and the orbit is precompact since it lives
in a finite dimension subspace. In the same way, if $\xi$ is of
the form $\xi=\oplus\xi_n$ with $\xi_n\neq0$ only for finitely
many indices $n$, we have an example of a theoretical quantum model with precompact orbits which are
not almost periodic.

\

\subsection{Quasienergy Operator and Almost Periodic Orbits}
Let $H(t)$ be a general time-dependent Hamiltonian in a Hilbert space
$\mathcal{H}$ such that the propagator
$U(t,s)$ is well defined. In this case we have defined the
quasienergy operator $K=-i\frac{d}{dt}+H(t)$ acting in the
extended Hilbert space $\mathcal{K}=L^2(\R,\mathcal{H},dt)$. It is
known~\cite{H,H1} that the quasienergy operator and the propagator
are connected by the relation
\begin{equation}\label{RelEq}
(e^{-iK\sigma}f)(t)=U(t,t-\sigma)f(t-\sigma).
\end{equation}
Let $\mathcal{K}_{\mathrm{p}}(K)$ and
$\mathcal{K}_{\mathrm{c}}(K)$ denote, respectively, the point
and continuous subspaces of $K$. We get the following result:

\

\begin{Proposition}\label{prop.31} Let $\xi\in\mathcal{H}$ be such that
$1\otimes\xi\in\mathcal{K}_{\mathrm{p}}(K)$. Then:
\begin{itemize}
\item[i)] The map $t\mapsto U(t,0)^{-1}\xi$ is almost periodic.
\item[ii)] If the eigenvectors
of $K$ have the form $\psi_m=1\otimes\xi_m$, with
$\xi_m\in\mathcal{H}$, then $\xi\in\mathcal{H}_{\mathrm{ap}}$.

\end{itemize}
\end{Proposition}
\begin{proof} If $1\otimes\xi\in\mathcal{K}_{\mathrm{p}}(K)$ then
$1\otimes\xi=\sum_mc_m\psi_m,$ with $K\psi_m=\lambda_m\psi_m$. So
\[e^{iK\sigma}(1\otimes\xi)=\sum_mc_me^{i\lambda_m\sigma}\psi_m,\]
therefore by (\ref{RelEq}) for each $t\in\R,$
\[U(t,t+\sigma)\xi=(e^{iK\sigma}(1\otimes\xi))(t)=
\sum_mc_me^{-i\lambda_m\sigma}\psi_m(t)\] and we conclude
that, for each fixed $t$, the map $\sigma\mapsto
U(t,t+\sigma)\xi$ is almost periodic. In particular taking $t=0$
we obtain that $\sigma\mapsto U(0,\sigma)\xi$ is almost periodic
and i) is proved.

Now, if the eigenvectors of $K$ have the form
$\psi_m=1\otimes\xi_m$, then
\begin{eqnarray*}
\xi(t)&=&U(t,0)\xi=(e^{-iKt}(1\otimes\xi))(t)\\&=&
\sum_mc_me^{-i\lambda_mt}\psi_m(t)\\&=&
\sum_mc_me^{-i\lambda_mt}\xi_m.
\end{eqnarray*} If the sum is finite the
map $t\mapsto\xi(t)$ is almost periodic since it is a
trigonometric polynomial. If the sum is infinite then
$\sum_{m=1}^kc_me^{-i\lambda_mt}\xi_m\rightarrow\sum_{m=1}^{\infty}c_me^{-i\lambda_mt}\xi_m$
uniformly as $k\rightarrow\infty$ and so the map
$t\mapsto\xi(t)$ is almost periodic, that is,
$\xi\in\mathcal{H}_{\mathrm{ap}}$, which is ii).
\end{proof}

\

\section{Bounded Energy}\label{BoundedSection}
In this section we consider time-dependent Hamiltonians
$H(t)=H_0+V(t)$ for which $H_0$ is a probe operator.

If $\psi_0\in\dom H_0$ and $\psi(t)=U(t,0)\psi_0$ is the solution
of the Schr\"{o}dinger equation, under which conditions
\[E_{\psi_0}^0(t)=\langle\psi(t),H_0\psi(t)\rangle\] is a bounded
function on $t$? Also, when
\[E_{\psi_0}(t)=\langle\psi(t),H(t)\psi(t)\rangle\] is a bounded
function? Next we present a set of simple general conditions related to the boundedness of such
energy functions.

\

\subsection{General Systems}
\begin{Proposition}\label{prop.38} If $V(t)$ is an uniformly bounded family of operators, that is,
$\sup_t\|V(t)\|<\infty$, then $E_{\psi_0}^0(t)$ is bounded if, and only if, $E_{\psi_0}(t)$ is
bounded.
\end{Proposition}
\begin{proof} It is sufficient to note that \[E_{\psi_0}(t) =
\langle\psi(t),H(t)\psi(t)\rangle =
E_{\psi_0}^0(t)+\langle\psi(t),V(t)\psi(t)\rangle\] and
\[\sup_t|\langle\psi(t),V(t)\psi(t)\rangle|\leq\sup_t\|\psi(t)\|^2\|V(t)\|=\sup_t
\|\psi_0\|^2\|V(t)\|<\infty.\]
\end{proof}

\begin{Proposition}\label{prop.39} If $\psi(t)\in
C^1(\R;\mathcal{H})$ is almost periodic and $\psi'(t)$ is
uniformly continuous, then $E_{\psi_0}(t)$ is bounded.
\end{Proposition}
\begin{proof} For each $n\in\N^*$ define
\[f_n(t)=n\left[\psi\left(t+\frac{1}{n}\right)-\psi(t)\right]=
n\int_{t}^{t+\frac{1}{n}}\psi'(s)ds.\] Since $\psi$ is almost
periodic it follows that $f_n$ is almost periodic for each $n$. As
$\psi'(t)$ is uniformly continuous, for each $\epsilon>0$ there
exists $\delta>0$ such that $|s-t|<\delta$ implies
$\|\psi'(t)-\psi'(s)\|<\epsilon$. Given $\epsilon>0$ let
$N(\epsilon)$ the smallest integer larger or equal to
$\frac{1}{\delta}$; then for all $n>N(\epsilon)$ and $t\in\R$
\begin{eqnarray*}
\left\|f_n(t)-\psi'(t)\right\| &=&
\left\|n\int_{t}^{t+\frac{1}{n}}(\psi'(s)-\psi'(t))ds\right\|\\ &\leq&
n\int_{t}^{t+\frac{1}{n}}\left\|\psi'(s)-\psi'(t)\right\|ds<\epsilon.
\end{eqnarray*}
So $f_n\rightarrow\psi'$ uniformly and therefore $\psi'(t)$ is
almost periodic.
Hence $i\psi'(t)$ and $\psi(t)$ are bounded
maps. Since
\[E_{\psi_0}(t)=\langle\psi(t),H(t)\psi(t)\rangle=
\langle\psi(t),i\frac{d\psi}{dt}(t)\rangle\] the result follows.
\end{proof}

\
Note that the boundedness of energy follows if $t\mapsto\psi(t)$ and $t\mapsto\psi'(t)$ are bounded
maps. Though well known, it is worth mentioning Proposition~\ref{prop.40} in this set of conditions.

\
\begin{Proposition}\label{prop.40} If $t\mapsto V(t)$ is strongly
$C^1$ and $\psi'(t)\in\dom H(t)$ for all $t$, then:
\begin{itemize}
\item[(a)] The map $t\mapsto E_{\psi}(t)$ is differentiable and
\[\frac{d}{dt}E_{\psi}(t)=\langle\psi(t),V'(t)\psi(t)\rangle.\]
\item[(b)] $\left|E_{\psi}(t)-E_{\psi}(0)\right|\leq t\times\sup_s\|V'(s)\|$.
\item[(c)] If there are $C>0, a>1$ so that  $\|V'(t)\|\leq\frac{C}{(1+|t|)^a}$, then
$E_{\psi}(t)$ and $E_{\psi}^0(t)$ are bounded functions.
\end{itemize}
\end{Proposition}
\begin{proof} (a) $E_{\psi}(t)=\langle\psi(t),(H_0+V(t))\psi(t)\rangle$ and so
\begin{eqnarray*}
\frac{d}{dt}E_{\psi}(t) &=& \langle\psi'(t),H(t)\psi(t)\rangle +
\langle\psi(t),H(t)\psi'(t)\rangle +
\langle\psi(t),V'(t)\psi(t)\rangle\\ &=&
\langle\psi'(t),i\psi'(t)\rangle + \langle i
\psi'(t),\psi'(t)\rangle + \langle\psi(t),V'(t)\psi(t)\rangle \\
&=& \langle\psi(t),V'(t)\psi(t)\rangle.
\end{eqnarray*}
(b) Since
\[E_{\psi}(t)-E_{\psi}(0)=\int_0^t\frac{d}{ds}E_{\psi}(s)ds=
\int_0^t\langle\psi(s),V'(s)\psi(s)\rangle ds\] the result
follows.
\newline (c) Similar to (b).
\end{proof}

A possibility for the proposition above is $V(t)=B_1\sin
t+\frac{B_2}{(1+|t|)^2}$ with $B_1,\;B_2\in B(\mathcal{H})$ and
self-adjoint. From this we see that certainly the choices of
$\psi$ depend on $B_1,B_2$, since $B_1\psi$ and $B_2\psi$ must be kept in suitable domains so that
$E_{\psi}(t)$ is meaningful.

\

\subsection{Purely Point Systems}
The next result is restricted to periodic time dependence and   Floquet
operators with nonempty point spectrum (see \cite{DSSV}).

\begin{Proposition}\label{prop.41} Let $V$ be periodic with period
$T$. If the subset $\{\xi_1,\ldots,\xi_n\}$ of eigenvectors
of $U_{\mathrm{F}}$ is in $\dom H_0$ and $t\mapsto\xi_j(t)$ are
$C^1$ maps, then for $\psi=\sum_{j=1}^na_j\xi_j$, where
$a_j\in\C,\;j=1,\cdots,n$, the map
$E_{\psi}(t)$ is bounded. If, moreover, $V(t)$ are bounded
operators and $\sup_t\|V(t)\|<\infty$, then $E_{\psi}^0(t)$ is also
bounded.
\end{Proposition}
\begin{proof} Suppose $U_{\mathrm{F}}\xi_j=e^{i\lambda_j}\xi_j$ with $\lambda_j\in\R$,
$1\leq j\leq n$. We have
\[E_{\xi_j,\xi_k}(t)\doteq\langle\xi_j(t),H(t)\xi_k(t)\rangle=
\langle\xi_j(t),i\frac{d}{dt}\xi_k(t)\rangle\] and so
$t\mapsto E_{\xi_j,\xi_k}(t)$ is continuous. Now
\begin{eqnarray*}
E_{\xi_j,\xi_k}(t+T) &=& \langle
U(t+T,0)\xi_j,H(t+T)U(t+T,0)\xi_k\rangle\\ &=& \langle
U(t+T,T)U_{\mathrm{F}}\xi_j,H(t)U(t+T,T)U_{\mathrm{F}}\xi_k\rangle\\ &=& e^{-i\lambda_j}
e^{i\lambda_k}\langle U(t,0)\xi_j,H(t)U(t,0)\xi_k\rangle\\ &=&
e^{i(\lambda_k-\lambda_j)}E_{\xi_j,\xi_k}(t)
\end{eqnarray*}
and then $t\mapsto E_{\xi_j,\xi_k}(t)$ is an almost periodic
function. Since for $\psi=\sum_{j=1}^na_j\xi_j$ we have
$E_{\psi}(t)=\sum_{j,k=1}^n\bar{a_j}a_kE_{\xi_j,\xi_k}(t)$ it follows that $E_{\psi}(t)$ is almost periodic
and so bounded. The second statement follows by Proposition~\ref{prop.38}.
\end{proof}

According to Proposition~\ref{prop.41}, in order to get dynamical stability in the
periodic case we need conditions assuring the eigenvectors of
$U_{\mathrm{F}}$ are in $\dom H_0$ and $t\mapsto\xi_j(t)$ to be
$C^1$ functions. We present some sufficient conditions in terms of the quasienergy operator $K$.

\

\begin{Lemma}\label{lema.42} Let $\xi\in\mathcal{H}$ be such that
$H(t)U(t,s)\xi$ is well defined. Then the map $t\mapsto
H(t)U(t,s)\xi$ is a $C^r$ function if, and only if, $t\mapsto
e^{i\lambda(t-s)}U(t,s)\xi$ is a $C^{r+1}$ function for fixed
$\lambda,\;s\in\R$.
\end{Lemma}
\begin{proof} Note that
\[\frac{d}{dt}(e^{i\lambda(t-s)}U(t,s)\xi)=i\lambda
e^{i\lambda(t-s)}U(t,s)\xi- ie^{i\lambda(t-s)}H(t)U(t,s)\xi.\]
Thus, if $t\mapsto H(t)U(t,s)\xi$ is $C^r$ then $t\mapsto
e^{i\lambda(t-s)}U(t,s)\xi$ is $C^{r+1}$ and reciprocally if
$t\mapsto e^{i\lambda(t-s)}U(t,s)\xi$ is $C^{r+1}$ then $t\mapsto
H(t)U(t,s)\xi$ is $C^r$.
\end{proof}

\

\begin{Corollary}\label{cor.43} If $f^{(\lambda)}$ is an eigenvector
of $K$, $Kf^{(\lambda)}=\lambda f^{(\lambda)}$, then the
map $t\mapsto f^{(\lambda)}(t)$ is $C^r$  if,
and only if, there exists $s\in\R$ so that $t\mapsto
H(t)U(t,s)f^{(\lambda)}(s)$ is $C^{r-1}$.
\end{Corollary}
\begin{proof} If $Kf^{(\lambda)}=\lambda f^{(\lambda)}$, then by
relation (\ref{RelEq}),
\[e^{-i\lambda\sigma}f^{(\lambda)}(t)=U(t,t-\sigma)f^{(\lambda)}(t-\sigma);\] so
$f^{(\lambda)}(t)=e^{i\lambda\sigma}U(t,t-\sigma)f^{(\lambda)}(t-\sigma)$
for all $\sigma\in\R$. Denoting $t-\sigma=s$ it follows that
$f^{(\lambda)}(t)=e^{i\lambda(t-s)}U(t,s)f^{(\lambda)}(s)$ and the
result follows by  Lemma~\ref{lema.42}.
\end{proof}

By using relation (\ref{RelEq}) one can easily show

\begin{Lemma}\label{Lemma12} For periodic systems with period
$T$, one has:
\begin{itemize}
\item[(a)] If $Kf=\lambda f$ then
$U_{\mathrm{F}}(s)f(s)=e^{-i\lambda T}f(s)$, $\forall\;s\in\R$.
\item[ (b)] If $U_{\mathrm{F}}(s)\xi_s=e^{-i\lambda T}\xi_s$,
$\xi_s\in\mathcal{H}$, $\forall\;s$, then
\[
f_{\xi}(t)\doteq \break e^{i\lambda(t-s)}U(t,s)\xi_s\in\dom K
\]
and $Kf_{\xi}=\lambda f_{\xi}$.
\end{itemize}
\end{Lemma}

\

\begin{Corollary}\label{cor.44} (a) If $H(t+T)=H(t)$, and $\xi^{(\lambda)}$
is an eigenvector of $U_{\mathrm{F}}(s)$,
$U_{\mathrm{F}}(s)\xi^{(\lambda)}=e^{-i\lambda T}\xi^{(\lambda)}$,
then $\xi^{(\lambda)}\in\dom H(s)$ if, and only if, there exists an eigenvector $f_{\xi^{(\lambda)}}$ of $K$,
$Kf_{\xi^{(\lambda)}}=\lambda f_{\xi^{(\lambda)}}$, with $t\mapsto
f_{\xi^{(\lambda)}}(t)$ continuous and differentiable.
\newline (b) In particular, $U_{\mathrm{F}}(s)$ has a basis of eigenvectors
in $\dom H(s)$ if, and only if, $K$ has a basis of eigenvectors
$\{f_j\}$ such that $t\mapsto f_j(t)$ is continuous and
differentiable for each $j$.
\end{Corollary}
\begin{proof} (a) Suppose that $\xi^{(\lambda)}\in\dom H(s)$. By
Lemma~\ref{Lemma12} $f_{\xi^{\lambda}}(t)=\break
e^{i\lambda(t-s)}U(t,s)\xi^{\lambda}\in\dom K$ and
$Kf_{\xi^{\lambda}}=\lambda f_{\xi^{\lambda}}$. Since
$\xi^{(\lambda)}\in\dom H(s)$ it follows that $U(t,s)\xi^{\lambda}\in\dom
H(t)$ and
$i\partial_tU(t,s)\xi^{\lambda}=H(t)U(t,s)\xi^{\lambda}$. Thus,
$t\mapsto f_{\xi^{(\lambda)}}(t)$ is continuous and
differentiable.

Reciprocally, it there exists an eigenvector $f_{\xi^{(\lambda)}}$  of
$K$ with $t\mapsto f_{\xi^{(\lambda)}}(t)$ continuous and
differentiable, then
$f_{\xi^{\lambda}}(t)=e^{i\lambda(t-s)}U(t,s)\xi^{\lambda}$ and
$Kf_{\xi^{(\lambda)}}=\lambda f_{\xi^{(\lambda)}}$ implies
$-i\partial_tf_{\xi^{\lambda}}(t)+H(t)f_{\xi^{\lambda}}(t)=\lambda
f_{\xi^{\lambda}}(t)$; therefore, $\xi^{\lambda}\in\dom H(s)$.

(b) It is a directly consequence of (a).
\end{proof}

\

\subsection{Jauslin-Lebowitz Formulation}
We want to study an analogue of the expectation value of probe operators $A:\dom
A\subset\mathcal{H}\rightarrow\mathcal{H}$ on the formulation presented by Jauslin and Lebowitz
\cite{JL,BJL} briefly recalled in the Introduction.  If the generalized
quasienergy  operator $\tilde{K}$ has pure point spectrum, there
exists an orthonormal basis $B\doteq\{f_n\}_{n=1}^{\infty}$ of
$\tilde{\mathcal{K}}$ with $\tilde{K}f_n=\lambda_nf_n$. By Theorem
4.2 in~\cite{JL}, if $f=1\otimes\xi$ is in the point subspace of
$\tilde{K}$ the function $t\mapsto U_{\theta}(t,0)\xi$ is almost
periodic a.e.\ $\theta$ with respect to the ergodic measure $\mu$ on the
compact manifold $\Omega$ (see Section~\ref{IntroductionSection}).

Denote \[B_{n,m}(A)\doteq\int_{\Omega}\left\langle
f_n(\theta),Af_m(\theta)\right\rangle_{\mathcal{H}}d\mu(\theta)=\left\langle
f_n,(1\otimes A)f_m\right\rangle_{\tilde{\mathcal{K}}}.\] If
$f\in\tilde{\mathcal{K}}$ then $f=\sum_na_nf_n$, with
$\sum_n|a_n|^2=\|f\|^2_{\tilde{\mathcal{K}}}$. For each time $t$, consider the average over $\Omega$ of
the expectation value of $A$, that is,
\begin{eqnarray*}
A_f(t) &\doteq& \int_{\Omega}\left\langle
U_{\theta}(t,0)f(\theta),AU_{\theta}(t,0)f(\theta)\right\rangle_{\mathcal{H}}d\mu(\theta)\\
&=& \int_{\Omega}\langle
(\mathcal{F}_te^{-i\tilde{K}t}f)(\theta),A(\mathcal{F}_te^{-i\tilde{K}t}f)(\theta)
\rangle_{\mathcal{H}}d\mu(\theta)
\end{eqnarray*}
\begin{eqnarray*}
&=& \left\langle\mathcal{F}_te^{-i\tilde{K}t}f,(1\otimes
A)\mathcal{F}_te^{-i\tilde{K}t}f\right\rangle_{\tilde{\mathcal{K}}}\\
&=& \left\langle e^{-i\tilde{K}t}f,(1\otimes
A)e^{-i\tilde{K}t}f\right\rangle_{\tilde{\mathcal{K}}}\\ &=&
\sum_{n,m}\overline{a_n}a_me^{-it(\lambda_m-\lambda_n)}\left\langle
f_n,(1\otimes A)f_m\right\rangle_{\tilde{\mathcal{K}}}\\ &=&
\sum_{n,m}\overline{a_n}a_me^{-it(\lambda_m-\lambda_n)}B_{n,m}(A).
\end{eqnarray*}

Note that if this sum is absolutely convergent then $A_f(t)$ is a
bounded and almost periodic function of $t$, and \[t\mapsto\left\langle
U_{\theta}(t,0)f(\theta),AU_{\theta}(t,0)f(\theta)\right\rangle_{\mathcal{H}}\]
is bounded a.e.\ $\theta$. We conclude

\

\begin{Proposition}\label{cor.45} If
$f=\sum_{j=1}^{m}a_jf_j$, where $f_j$ are eigenvectors of
$\tilde{K}$ and $f_j(\theta)\in\dom A$, for all $\theta$, then
$t\mapsto A_f(t)$ is a bounded and almost periodic function.
Moreover,
\[t\mapsto\left\langle
U_{\theta}(t,0)f(\theta),AU_{\theta}(t,0)f(\theta)\right\rangle_{\mathcal{H}}\]
is bounded for almost every $\theta$.
\end{Proposition}

More generally we obtain the following result:

\

\begin{Theorem}\label{teo.47} Suppose that $\Omega$ is a compact
manifold, $g_t:\Omega\rightarrow\Omega$ a $C^1$ flow  with
$\sup_{t,\theta}\|\partial_tg_t(\theta)\|<\infty$, and
$\tilde{K}f^{(\lambda)}=\lambda f^{(\lambda)}$ with $\theta\mapsto
f^{(\lambda)}(\theta)$ a $C^1$ map. Then for $\mu$ almost every $\theta$ one has
$U_{\theta}(t,0)f^{(\lambda)}(\theta)\in\dom H_{\theta}(t)$ and
\[\left\langle U_{\theta}(t,0)f^{(\lambda)}(\theta),H_{\theta}(t)
U_{\theta}(t,0)f^{(\lambda)}(\theta)\right\rangle\] is a bounded
function of $t$. Moreover, if $H_{\theta}(t)=H_0+V(g_t(\theta))$
with $V(g_t(\theta))$ bounded and
$\sup_{t,\theta}\|V(g_t(\theta))\|<\infty$, then the energy expectation
 \[\left\langle
U_{\theta}(t,0)f^{(\lambda)}(\theta),H_0
U_{\theta}(t,0)f^{(\lambda)}(\theta)\right\rangle\] is also bounded.
\end{Theorem}
\begin{proof} Since $\tilde{K}f^{(\lambda)}=\lambda f^{(\lambda)}$ then
$f^{(\lambda)}(\theta)\in\dom H_{\theta}(0)$ a.e.\ $\theta$ and
therefore $U_{\theta}(t,0)f^{(\lambda)}(\theta)\in\dom
H_{\theta}(t)$ a.e.\ $\theta$. On the other hand
\[U_{\theta}(t,0)f^{(\lambda)}(\theta)=
\mathcal{F}_te^{-i\tilde{K}t}f^{(\lambda)}(\theta)=
\mathcal{F}_te^{-i\lambda t}f^{(\lambda)}(\theta)=e^{-i\lambda
t}f^{(\lambda)}(g_t(\theta))\] and from the differentiability
hypothesis it follows that
\[i\frac{\partial}{\partial
t}U_{\theta}(t,0)f^{(\lambda)}(\theta)=\lambda e^{-i\lambda
t}f^{(\lambda)}(g_t(\theta))+ie^{-i\lambda
t}\frac{d}{d\theta}f^{(\lambda)}(g_t(\theta))\frac{d}{dt}g_t(\theta),\]
which implies that
\[
i\frac{\partial}{\partial
t}U_{\theta}(t,0)f^{(\lambda)}(\theta)=H_{\theta}(t)
U_{\theta}(t,0)f^{(\lambda)}(\theta)
\] is bounded and the first
part of the result is proved. The second one follows as in Proposition~\ref{prop.38}.
\end{proof}

\

\begin{Corollary}\label{cor.48} Suppose the hypotheses of the above theorem
hold  and that for each eigenvector
$f^{(\lambda_n)}\in\tilde{\mathcal{K}}$ the function
$\theta\mapsto f^{(\lambda_n)}(\theta)$ is $C^1$. Then for $\mu$ almost every $\theta$ and for all
vectors
$\xi\in\mathcal{H}$ of the form
\[\xi=a_1f^{(\lambda_1)}(\theta)+\ldots+a_kf^{(\lambda_k)}(\theta),\]
the expectation value of the energy
\[\left\langle
U_{\theta}(t,0)\xi,H_{\theta}(t) U_{\theta}(t,0)\xi\right\rangle\] is a
bounded function.
\end{Corollary}

\

In case $\xi=\sum_{n=1}^{\infty}a_nf^{(\lambda_m)}(\theta)$
with $\sum|a_n|^2<\infty$, a sufficient condition
for $U_{\theta}(t,0)\xi\in\dom H_{\theta}(t)$ and bounded energy
is \[\sum_{j=1}^{\infty}|a_j|\left(|\lambda_j|+
\sup_{\theta}\|\partial_{\theta}f^{\lambda_j}(\theta)\|\right)<\infty,\]
since this implies that \[t\mapsto
U_{\theta}(t,0)\xi=\sum_{j=1}^{\infty}a_je^{-i\lambda_jt}f^{\lambda_j}(g_t(\theta))\]
is a $C^1$ function and $i\partial_tU_{\theta}(t,0)$ is bounded.

\end{document}